\documentclass[submission,copyright,creativecommons]{eptcs}
\providecommand{\event}{ACT2021} % Name of the event you are submitting to
\usepackage{breakurl}       % Not needed if you use pdflatex only.
\usepackage{underscore}      % Only needed if you use pdflatex.

\usepackage{enumerate,xspace}
\usepackage{amsmath,amssymb,wasysym}
\usepackage[all]{xy}
\usepackage{proof}
\usepackage[svgnames]{xcolor}
\usepackage{tikz}
\usepackage[sort,nocompress]{cite}
\usepackage{stmaryrd} %need for special interpretation bracket
\usepackage{mathtools}
\usepackage{latexsym}
\usepackage{dsfont}
\usepackage{multicol}
\usepackage{cmll}
\usepackage{multirow}
\usepackage{longtable}

\newtheorem{observation}{Remark}[section]
\newtheorem{definition}[observation]{Definition}
\newtheorem{example}[observation]{Example}

\newtheorem{proposition}[observation]{Proposition}

\title{Jacobians and Gradients \\for Cartesian Differential Categories}
\author{Jean-Simon Pacaud Lemay\thanks{The author would like to thank Jonathan Gallagher for useful discussions. The author is financially supported by a Natural Science and Engineering Research Council of Canada (NSERC) Postdoctoral Fellowship (PDF) - Award \#: 456414649}
\institute{Mathematics and Computer Science Department\\ Mount Allison University\\ Sackville, New Brunswick, Canada}
\email{jsplemay@gmail.com} \footnote{Author's website: \url{https://sites.google.com/view/jspl-personal-webpage/home}} }
\def\titlerunning{Jacobians and Gradients for CDCs}
\def\authorrunning{J.-S. P. Lemay}
\begin{document}
\maketitle

\begin{abstract} Cartesian differential categories come equipped with a differential combinator that formalizes the directional derivative from multivariable calculus. Cartesian differential categories provide a categorical semantics of the differential $\lambda$-calculus and have also found applications in causal computation, incremental computation, game theory, differentiable programming, and machine learning. There has recently been a desire to provide a (coordinate-free) characterization of Jacobians and gradients in Cartesian differential categories. One's first attempt might be to consider Cartesian differential categories which are Cartesian closed, such as models of the differential $\lambda$-calculus, and then take the curry of the derivative. Unfortunately, this approach excludes numerous important examples of Cartesian differential categories such as the category of real smooth functions. In this paper, we introduce linearly closed Cartesian differential categories, which are Cartesian differential categories that have an internal hom of linear maps, a bilinear evaluation map, and the ability to curry maps which are linear in their second argument. As such, the Jacobian of a map is defined as the curry of its derivative. Many well-known examples of Cartesian differential categories are linearly closed, such as, in particular, the category of real smooth functions. We also explain how a Cartesian closed differential category is linearly closed if and only if a certain linear idempotent on the internal hom splits. To define the gradient of a map, one must be able to define the transpose of the Jacobian, which can be done in a Cartesian reverse differential category. Thus, we define the gradient of a map to be the curry of its reverse derivative and show this equals the transpose of its Jacobian. We also explain how a linearly closed Cartesian reverse differential category is precisely a linearly closed Cartesian differential category with an appropriate notion of transpose. 
\end{abstract}

\section{Introduction}

Cartesian differential categories, introduced by Blute, Cockett, and Seely in \cite{blute2009cartesian}, come equipped with a differential combinator $\mathsf{D}$ which provides a categorical axiomatization of the directional derivative from multivariable calculus, and so for every map $A \xrightarrow{f} B$ produces its derivative $A \times A \xrightarrow{\mathsf{D}[f]} B$. There is no shortage of examples of Cartesian differential categories in the literature, but arguably the most important example is $\mathsf{SMOOTH}$, the category of Euclidean spaces $\mathbb{R}^n$ and real smooth functions between them, where the differential combinator, in this case, is precisely the classical directional derivative. Another important class of examples are the Cartesian closed differential categories \cite{Cockett-2019,bucciarelli2010categorical,manzonetto2012}, which provide a categorical semantics of Ehrhard and Regnier's differential $\lambda$-calculus \cite{ehrhard2003differential}. Cartesian (closed) differential categories have found numerous applications in computer science such as being picked up by Katsumata and Sprunger in their work on causal computations \cite{sprunger2019differentiable}, by Abadi and Plotkin in their work on differentiable programming languages \cite{abadi2019simple}, by Alvarez-Picallo and Ong in their work on incremental computation \cite{alvarez2019change},  and by Laird, Manzonetto, and McCusker in their work in game theory \cite{laird2013constructing}. Most recently, Cartesian (closed) differential categories have also found usage in machine learning with the introduction of Cartesian \emph{reverse} differential categories \cite{cockettetal:LIPIcs:2020:11661,cruttwell2020categorical}, which have been shown to be a suitable setting for reverse gradient descent by Cruttwell, Gavranovi{\'c}, Ghani, Wilson, and Zanasi in \cite{cruttwell2021categorical, wilson2021reverse}. 

An important concept in differential calculus is the Jacobian matrix. Recall that for a smooth function $\mathbb{R}^n \xrightarrow{F = \langle f_1, \hdots, f_n \rangle} \mathbb{R}^m$, its Jacobian matrix at point $\vec x \in \mathbb{R}^n$ is given by the $n \times m$ matrix $\mathbf{J}(F)(\vec x)$ whose coordinates are the partial derivatives of $f_i$ evaluated at $\vec x$. The $\mathbb{R}$-linear function associated with the Jacobian matrix at $\vec x$ is called the total derivative of $F$, and evaluating this linear function at $\vec y$ results in the derivative $\mathsf{D}[F](\vec x, \vec y)$. Thus, the Jacobian of $F$ can be interpreted as a map $\mathbb{R}^n \xrightarrow{\mathbf{J}(F)} \mathsf{LIN}(\mathbb{R}^n, \mathbb{R}^m)$, where the codomain is the vector spaces of $\mathbb{R}$-linear functions from $\mathbb{R}^n$ to $\mathbb{R}^m$, and can therefore be understood as sort of curry $\mathbb{R}^n \times \mathbb{R}^n \xrightarrow{\mathsf{D}[F]} \mathbb{R}^m$. Interpreting the Jacobian in this manner shows how the Jacobian is a special case of both the Fr{\'e}chet derivative and the Gateaux derivative. For a smooth function $\mathbb{R}^n \xrightarrow{f} \mathbb{R}$, its gradient at $\vec x \in \mathbb{R}^n$ is the transpose of its Jacobian at $\vec x$, $\nabla(f)(\vec x) = \mathbf{J}(F)(\vec x)^\mathsf{T} \in \mathbb{R}^n$. Since $\mathsf{LIN}(\mathbb{R}, \mathbb{R}^n) \cong \mathbb{R}^n$, the gradient can therefore be interpreted as a map $\mathbb{R}^n \xrightarrow{\nabla(f)} \mathsf{LIN}(\mathbb{R}, \mathbb{R}^n)$. 

The notions of Jacobians and gradients have yet to be formally defined in a Cartesian differential category. Such concepts would be highly desirable, specifically if one wishes to formalize machine learning algorithms in Cartesian differential categories. This need for Jacobians and gradients is expressed by Katsumata and Sprunger, who state the following in the conclusion of their paper \cite{sprunger2019differentiable}: ``Though we would like to say our abstract treatment of differentiation can be used directly by machine learning practitioners, it appears this is not the case yet. The derivative of a morphism in a Cartesian differential category is not the same as having an explicit Jacobian or gradient. A gradient can be recovered from this morphism by applying it to all the basis vectors, but when there are millions of parameters in a machine learning model, this idea is computationally disastrous. We think that by adding some structure to Cartesian differential categories, such as a designated closed subcategory, we could give a theoretical treatment allowing for more explicit representation of Jacobians.'' This highlights a clear need for axiomatizing (coordinate-free) Jacobians and gradients in the context of Cartesian differential categories. 

Based on the above discussion, and as suggested by Katsumata and Sprunger, one's first attempt might be to consider defining Jacobians in a Cartesian closed differential category and define the Jacobian of a map $A \xrightarrow{f} B$ as the curry of its derivative $A \times A \xrightarrow{\mathsf{D}[f]} B$, so $A \xrightarrow{\lambda\left( \mathsf{D}[f] \right)} [A,B]$. While this is a very reasonable and promising idea that enables one to one can get quite far in extracting the main properties of the Jacobian, there is a flaw. The Cartesian closed differential category approach would exclude numerous examples of Cartesian differential categories, specifically those wfith a ``finite-dimensional flavour'' and many machine learning-related models. In particular, $\mathsf{SMOOTH}$ is not Cartesian closed since, intuitively, the set of smooth functions from $\mathbb{R}^n$ to $\mathbb{R}^m$ is not a finite-dimensional $\mathbb{R}$-vector space and therefore not isomorphic to a Euclidean space $\mathbb{R}^k$. Thus another approach is required. Luckily, the codomain of the Jacobian is $\mathsf{LIN}(\mathbb{R}^n, \mathbb{R}^m)$, which is isomorphic to $\mathbb{R}^{nm}$, and so can therefore be interpreted as an object in $\mathsf{SMOOTH}$. We may generalize this idea in numerous Cartesian differential categories since there is a notion of linearity based on the differential combinator (which often coincides with the notion of linearity from linear algebra). Thus, one needs the notion of internal \emph{linear} homs. 

In this paper, we introduce the notion of a \textbf{linearly closed Cartesian differential category}, which is a Cartesian differential category with internal linear hom $\mathcal{L}(A,B)$ (which can be interpreted as an object which represents the set of linear maps from $A$ to $B$), a bilinear evaluation map $\mathcal{L}(A,B) \times A \xrightarrow{\epsilon_\ell} B$, and the ability to curry maps which are linear in their second argument. Many analogues of the basic Cartesian closed properties hold for the linearly closed setting. One of the axioms of the differential combinator states precisely that for any map $A \xrightarrow{f} B$, its derivative $A \times A \xrightarrow{\mathsf{D}[f]} B$ is linear in its second argument. Therefore, the Jacobian of $f$ is defined as the curry of $\mathsf{D}[f]$, $A \xrightarrow{\mathbf{J}(f) : = \lambda_\ell( \mathsf{D}[f] )} \mathcal{L}(A,B)$, and numerous of the basic properties of the Jacobian from classical calculus hold. Many important examples of Cartesian differential categories are linearly closed, such as any differential Lawvere theory like $\mathsf{SMOOTH}$ or the Lawvere theory of polynomials over a commutative semiring, and also the coKleisli category of a differential category which is symmetric monoidal closed \cite{blute2006differential}. We also explain how a Cartesian closed differential category is linearly closed if and only if a certain linear idempotent on the internal hom splits. Therefore, a certain idempotent completion of a Cartesian closed differential category results in a linearly closed Cartesian differential category. Lastly, to define the gradient of a map, one requires the ability to take the transpose (or dagger) of linear maps. This can be achieved in a linearly closed Cartesian reverse differential category, which we explain is precisely a linear closed Cartesian differential category equipped with linear transpose maps $\mathcal{L}(A,B) \xrightarrow{\tau} \mathcal{L}(B,A)$. Therefore, the gradient of a map $A \xrightarrow{f} B$ is defined by taking the curry of its reverse derivative $A \times B \xrightarrow{\mathsf{R}[f]} A$, which is linear in its second argument, $A \xrightarrow{\nabla(f) : = \lambda_\ell( \mathsf{R}[f] )} \mathcal{L}(B,A)$, which we show is equal to the transpose of its Jacobian, so $\nabla(f) = \tau \circ \mathbf{J}(f)$. In future works, the notions of Jacobians and gradient will be particularly useful when generalizing and applying automatic differentiation and machine learning algorithms, such as back-propagation or (reverse) gradient descent, in the setting of a Cartesian (reverse) differential category. In fact, internal linear homs are a key concept in V\'ak\'ar's recent work on automatic differentiation \cite{vakar2021chad}.

\section{Background: Cartesian Differential Categories and Linear Maps} 

In this background section, we review Cartesian (closed) differential categories and linear maps, as well as providing examples and a term logic notation that will help with intuition. For a more in-depth introduction to Cartesian differential categories, we refer the reader to \cite{blute2009cartesian, cockett2020linearizing}.

The underlying structure of a Cartesian differential category is that of a Cartesian left additive category, which in particular allows one to have zero maps and sums of maps, while also allowing for maps that do not preserve said sums or zeros. Maps that do preserve the additive structure are called \emph{additive} maps. It is important to note that we do not assume that our Cartesian left additive categories necessarily have negatives, which allows for examples from computer science. For a category with (chosen) finite products we denote the (chosen) binary product as $\times$, with projection maps $A \times B \xrightarrow{\pi_0} A$ and ${A \times B \xrightarrow{\pi_1} B}$ and pairing object $\langle -, - \rangle$, and the (chosen) terminal object as $\top$. 

\begin{definition}\label{CLACdef} A \textbf{left additive category} \cite[Definition 1.1.1]{blute2009cartesian} is a category $\mathbb{X}$ such that each hom-set $\mathbb{X}(A,B)$ is a commutative monoid with addition $\mathbb{X}(A,B) \times \mathbb{X}(A,B) \xrightarrow{+} \mathbb{X}(A,B)$, $(f,g) \mapsto f +g$, and zero $0 \in \mathbb{X}(A,B)$, and such that pre-composition preserves the additive structure: $(f+g) \circ a = f \circ a + g \circ a$ and $0 \circ a = 0$. A map $A \xrightarrow{f} B$ is said to be \textbf{additive} \cite[Definition 1.1.1]{blute2009cartesian} if post-composition by $f$ preserves the additive structure: $f \circ (a + b) = f \circ a + g \circ b$ and $f \circ 0 = 0$. A \textbf{Cartesian left additive category} \cite[Definition 1.2.1]{blute2009cartesian} is a left additive category $\mathbb{X}$ which has finite products and such that all the projection maps $A \times B \xrightarrow{\pi_0} A$ and ${A \times B \xrightarrow{\pi_1} B}$ are additive.\footnote{We note that the definition of a Cartesian left additive category presented here is not precisely that given in \cite[Definition 1.2.1]{blute2009cartesian}, but it is indeed equivalent, as explained in \cite{cockett2020linearizing}.}\end{definition}   

Cartesian differential categories are Cartesian left additive categories that come equipped with a differential combinator, whose axioms capture the basic properties of the directional derivative from multivariable differential calculus. In the following definition, it is important to note that we follow the convention used in more recent work on Cartesian differential categories which flips the convention used in \cite{blute2009cartesian}, so that the linear argument of the derivative is in the second argument rather than in the first. 

\begin{definition}\label{cartdiffdef} A \textbf{Cartesian differential category} \cite[Definition 2.1.1]{blute2009cartesian} is a Cartesian left additive category $\mathbb{X}$ equipped with a \textbf{differential combinator} $\mathsf{D}$, which is a family of operators $\mathbb{X}(A,B) \xrightarrow{\mathsf{D}} \mathbb{X}(A \times A,B)$, where for a map $A \xrightarrow{f} B$, the resulting map $A \times A \xrightarrow{\mathsf{D}[f]} B$ is called the derivative of $f$, and such that:  
\begin{enumerate}[{\bf [CD.1]}]
\item \label{CDCax1} $\mathsf{D}[f+g] = \mathsf{D}[f] + \mathsf{D}[g]$ and $\mathsf{D}[0] = 0$ 
\item \label{CDCax2} ${\mathsf{D}[f] \! \circ \! \langle a, b +c \rangle \!=\! \mathsf{D}[f] \!\circ\! \langle a, b \rangle \!+\! \mathsf{D}[f] \!\circ\! \langle a, c \rangle}$ and $\mathsf{D}[f] \circ \langle a, 0 \rangle = 0$
\item \label{CDCax3} $\mathsf{D}[1_A]=\pi_1$ and $\mathsf{D}[\pi_j] = \pi_j \circ \pi_1$
\item \label{CDCax4} $\mathsf{D}[\left\langle f,g \right \rangle] = \left \langle \mathsf{D}[f], \mathsf{D}[g] \right \rangle$
\item \label{CDCax5} $\mathsf{D}[g \circ f] = \mathsf{D}[g] \circ \langle f \circ \pi_0, \mathsf{D}[f] \rangle$
\item \label{CDCax6} $\mathsf{D}\left[\mathsf{D}[f] \right] \circ \left \langle \langle a,b \rangle, \langle 0, c \rangle \right \rangle = \mathsf{D}[f] \circ \langle a,c \rangle$
\item \label{CDCax7} ${\mathsf{D}\left[\mathsf{D}[f] \right] \circ \left \langle \langle a,b \rangle, \langle c, 0 \rangle \right \rangle = \mathsf{D}\left[\mathsf{D}[f] \right] \circ \left \langle \langle a,c \rangle, \langle b, 0 \rangle \right \rangle}$
\end{enumerate}
\end{definition}

To help with the intuition, it is useful to use the term logic\index{term logic} of Cartesian differential categories as introduced in \cite[Section 4]{blute2009cartesian}, which expresses the differential combinator as: $\mathsf{D}[f](a,b) := \frac{\mathsf{d}f(x)}{\mathsf{d}x}(a) \cdot b$. {\bf [CD.1]} says that the derivative of a sum is equal to the sum of the derivatives, $\frac{\mathsf{d}f(x)+g(x)}{\mathsf{d}x}(a) \cdot b = \frac{\mathsf{d}f(x)}{\mathsf{d}x}(a) \cdot b + \frac{\mathsf{d}g(x)}{\mathsf{d}x}(a) \cdot b$, and that the derivative of zero maps is zero, $\frac{\mathsf{d}0}{\mathsf{d}x}(a) \cdot b = 0$. {\bf [CD.2]} says that derivatives are additive in their second argument, so that $\frac{\mathsf{d}f(x)}{\mathsf{d}x}(a) \cdot (b+c) = \frac{\mathsf{d}f(x)}{\mathsf{d}x}(a) \cdot b + \frac{\mathsf{d}f(x)}{\mathsf{d}x}(a) \cdot c$ and $\frac{\mathsf{d}f(x)}{\mathsf{d}x}(a) \cdot 0 =0$. {\bf [CD.3]} tells us what the derivatives of identity maps and projections maps are, so that $\frac{\mathsf{d}x}{\mathsf{d}x}(a) \cdot b = b$ and $\frac{\mathsf{d}\pi_i(x_0,x_1)}{\mathsf{d}(x_0,x_1)}(a_0,a_1) \cdot (b_0,b_1) = b_i$. {\bf [CD.4]} says the derivative of a pairing of maps is equal to the pairing of the derivatives, $\frac{\mathsf{d}\langle f(x), g(x) \rangle}{\mathsf{d}x}(a) \cdot b = \left \langle \frac{\mathsf{d}f(x)}{\mathsf{d}x}(a) \cdot b, \frac{\mathsf{d}g(x)}{\mathsf{d}x}(a) \cdot b \right \rangle$. {\bf [CD.5]} is the chain rule which tells us what the derivative of a composition of functions is, so $\frac{\mathsf{d}g\left( f(x) \right)}{\mathsf{d}x}(a) \cdot b = \frac{\mathsf{d}g(y)}{\mathsf{d}y}(f(a)) \cdot \left( \frac{\mathsf{d}f(x)}{\mathsf{d}x}(a) \cdot b \right) $. The last two axioms {\bf [CD.6]} and {\bf [CD.7]} may look somewhat mysterious but essentially they capture properties of partial differentiation. Indeed, in any Cartesian differential category, one can define partial differentiation by inserting zeros in the total differential \cite[Section 4.5]{blute2009cartesian}, which in the term logic is written respectively as: $\frac{\mathsf{d}f(x,b)}{\mathsf{d}x}(a) \cdot c := \frac{\mathsf{d}f(x,y)}{\mathsf{d}\langle x, y \rangle}(a,b) \cdot (c,0)$ and $ \frac{\mathsf{d}f(a,y)}{\mathsf{d}y}(b) \cdot c := \frac{\mathsf{d}f(x,y)}{\mathsf{d}\langle x, y \rangle}(a,b) \cdot (0,c)$. Thus, {\bf [CD.6]} tells us what the partial derivative in the second argument of a derivative is, $\frac{\mathsf{d}\frac{\mathsf{d}f(x)}{\mathsf{d}x}(a) \cdot z}{\mathsf{d}z}(b) \cdot c = \frac{\mathsf{d}f(x)}{\mathsf{d}x}(a) \cdot c$, while {\bf [CD.7]} captures the symmetry of the partial derivatives, $\frac{\mathsf{d}\frac{\mathsf{d}f(x)}{\mathsf{d}x}(y) \cdot b}{\mathsf{d}(y)}(a) \cdot c = \frac{\mathsf{d}\frac{\mathsf{d}f(y)}{\mathsf{d}y}(x) \cdot c}{\mathsf{d}(x)}(a) \cdot b$. We note that partial derivatives are also differential combinators for the simple slice categories \cite[Corollary 4.5.2]{blute2009cartesian}, and the total derivative is equal to the sum of the partial derivatives \cite[Lemma 4.5.1]{blute2009cartesian}, that is, $\frac{\mathsf{d}f(x,y)}{\mathsf{d}\langle x, y \rangle}(a,b) \cdot (c,d) = \frac{\mathsf{d}f(x,b)}{\mathsf{d}x}(a) \cdot c + \frac{\mathsf{d}f(a,y)}{\mathsf{d}y}(b) \cdot d$. More discussion on the differential combinator axioms can be found in \cite[Remark 2.1.3]{blute2009cartesian}. 

Here are now our main examples of Cartesian differential categories (see \cite{cockett2020linearizing} for a list of more examples of Cartesian differential categories): 

\begin{example}\label{ex:biproduct} \normalfont Let $R$ be a commutative ring and let $R\text{-}\mathsf{MOD}$ be the category of $R$-modules and $R$-linear maps between them. $R\text{-}\mathsf{MOD}$ is a Cartesian differential category where for an $R$-linear map $M \xrightarrow{f} N$, its derivative $M \times M \xrightarrow{\mathsf{D}[f]} N$ is defined as $\mathsf{D}[f](x,y) = f(y)$. More generally, every category with finite biproducts is a Cartesian differential category where ${\mathsf{D}[f] = f \circ \pi_1}$. 
\end{example}

\begin{example}\label{ex:smooth} \normalfont Let $\mathbb{R}$ be the set of real numbers. Define $\mathsf{SMOOTH}$ as the category whose objects are the Euclidean real vector spaces $\mathbb{R}^n$ and whose maps are the real smooth functions ${\mathbb{R}^n \xrightarrow{F} \mathbb{R}^m}$ between them. $\mathsf{SMOOTH}$ is a Cartesian differential category where the differential combinator is defined as the directional derivative of a smooth function. Recall that a smooth function $\mathbb{R}^n \xrightarrow{F} \mathbb{R}^m$ is in fact a tuple $F = \langle f_1, \hdots, f_m \rangle$ of smooth functions $\mathbb{R}^n \xrightarrow{f_i} \mathbb{R}$. Then using the convention that $\vec x \in \mathbb{R}^n$ are column vectors, the derivative $\mathbb{R}^n \times \mathbb{R}^n \xrightarrow{\mathsf{D}[F]} \mathbb{R}^m$ is defined as multiplying the Jacobian matrix of $F$ at the first argument $\vec x$, which is an $m \times n$ matrix $\mathbf{J}(F)(\vec x)$, with the second argument $\vec y$, seen as an $n \times 1$ matrix:
\[{\small \mathsf{D}[F](\vec x, \vec y) := \mathbf{J}(F)(\vec x)\vec y = \begin{bmatrix} \frac{\partial f_1}{\partial x_1}(\vec x) & \frac{\partial f_1}{\partial x_2}(\vec x) & \hdots & \frac{\partial f_1}{\partial x_n}(\vec x) \\
 \frac{\partial f_2}{\partial x_1}(\vec x) & \frac{\partial f_2}{\partial x_2}(\vec x) & \hdots & \frac{\partial f_2}{\partial x_n}(\vec x) \\
 \vdots & \vdots & \vdots & \vdots \\
 \frac{\partial f_m}{\partial x_1}(\vec x) & \frac{\partial f_m}{\partial x_2}(\vec x) & \hdots & \frac{\partial f_m}{\partial x_n}(\vec x) 
\end{bmatrix} \begin{bmatrix} y_1 \\ y_2 \\ \vdots \\ y_n
\end{bmatrix} = \begin{bmatrix} \sum \limits^n_{i=1} \frac{\partial f_1}{\partial x_i}(\vec x) y_i \\ \vdots \\ \sum \limits^n_{i=1} \frac{\partial f_m}{\partial x_i}(\vec x) y_i
\end{bmatrix} } \]
When $m=1$, for a smooth function $\mathbb{R}^n \xrightarrow{f} \mathbb{R}$, $\mathsf{D}[f](\vec x, \vec y)$ is precisely the directional derivative of $f$ at point $\vec x$ and along the vector $\vec y$. 
\end{example}

\begin{example} \normalfont An important source of examples of Cartesian differential categories is the coKleisli categories of differential categories \cite{blute2006differential}. Very briefly (and leaving out most of the details), a differential category \cite[Definition 2.4]{blute2006differential} is a symmetric monoidal category (with tensor product $\otimes$ and unit $k$) with a comonad $\oc$ which comes equipped with a deriving transformation $\oc A \otimes A \xrightarrow{\mathsf{d}} A$ satisfying certain coherences which capture the basic algebraic properties of differentiation \cite[Definition 7]{Blute2019}. Examples of differential categories can be found in \cite[Section 9]{Blute2019}. By \cite[Proposition 3.2.1]{blute2009cartesian}, when a differential category has finite products, the coKleisli category of $\oc$ is a Cartesian differential category where the differential combinator is defined using the deriving transformation. 
\end{example}

An important class of maps in a Cartesian differential category are the linear maps and maps which are linear in certain arguments. Essentially, a map is linear in an argument if when differentiating with respect to that argument (and keeping the other arguments constant), one gets back the starting map. 

\begin{definition} In a Cartesian differential category $\mathbb{X}$ with differential combinator $\mathsf{D}$:
\begin{enumerate}[{\em (i)}]
\item A map $A \xrightarrow{f} B$ is \textbf{linear} \cite[Definition 2.2.1]{blute2009cartesian} if $\mathsf{D}[f] \circ \langle a, b \rangle = f \circ b$ (i.e.\ $\frac{\mathsf{d}f(x)}{\mathsf{d}x}(a) \cdot b = f(b)$);
\item A map $A \times B \xrightarrow{f} C$ is \textbf{linear in its first argument} if $f$ is linear with respect to the partial derivative in its first argument, that is, $\mathsf{D}[f] \circ \left \langle \langle a,b \rangle, \langle c, 0 \rangle \right \rangle = f \circ \langle c, b \rangle$ (i.e.\ ${\frac{\mathsf{d}f(x,b)}{\mathsf{d}x}(a) \cdot c = f(c,b)}$); 
\item A map $A \times B \xrightarrow{f} C$ is \textbf{linear in its second argument} if $f$ is linear with respect to the partial derivative in its second argument, that is, $\mathsf{D}[f] \circ \left \langle \langle a,b \rangle, \langle 0, d \rangle \right \rangle = f \circ \langle a, d \rangle$ (i.e.\ $\frac{\mathsf{d}f(a,y)}{\mathsf{d}y}(b) \cdot d = f(a,d)$); 
\item A map $A \times B \xrightarrow{f} C$ is \textbf{bilinear} if it is linear in its first argument and linear in its second argument, or equivalently if $\mathsf{D}[f] \circ \left \langle \langle a,b \rangle, \langle c, d \rangle \right \rangle = f \circ \langle a, d \rangle + f \circ \langle c,b \rangle$ (i.e.\ ${\frac{\mathsf{d}f(x,y)}{\mathsf{d}( x, y)}(a,b) \cdot (c,d) = f(a,d) + f(c,b)}$). 
\end{enumerate}
Define the subcategory of linear maps $\mathsf{Lin}[\mathbb{X}]$ to be the category whose objects are the same as $\mathbb{X}$ and whose maps are linear in $\mathbb{X}$, and let $\mathsf{U}: \mathsf{Lin}[\mathbb{X}] \to \mathbb{X}$ be the obvious inclusion functor. 
\end{definition}

By {\bf [CD.6]}, for any map $f$, its derivative $\mathsf{D}[f]$ is linear in its second argument, while by {\bf [CD.3]}, identity maps and projections are linear. In fact, by \cite[Lemma 2.2.2]{blute2009cartesian}, linear maps are also closed under composition, addition, and pairings, which implies that $\mathsf{Lin}[\mathbb{X}]$ is a Cartesian left additive category with the same structure as $\mathbb{X}$. Furthermore, every linear map is also additive (though the converse is not always true), and as a result, the finite product structure of $\mathbb{X}$ becomes a finite biproduct structure in $\mathsf{Lin}[\mathbb{X}]$, where in particular, the injection maps are defined as $A \xrightarrow{\iota_0 := \langle 1_A,0 \rangle} A \times B$ and $B \xrightarrow{\iota_1 := \langle 0, 1_B \rangle} A \times B$. 

\begin{example} \normalfont In $R\text{-}\mathsf{MOD}$, every $R$-linear map is linear in the Cartesian differential sense, so we have that $\mathsf{Lin}[R\text{-}\mathsf{MOD}] = R\text{-}\mathsf{MOD}$. On the other hand, $M \times M^\prime \xrightarrow{f} N$ is linear in its first argument (resp.\ second argument) if and only if $f(x,y) = f(x,0)$ (resp.\ $f(x,y) = f(0,y)$). Therefore, the only maps which are bilinear in the Cartesian differential sense are the zero maps. The same story holds true for any category with finite biproducts. 
\end{example}

\begin{example} \normalfont In $\mathsf{SMOOTH}$, a smooth function $\mathbb{R}^n \xrightarrow{F} \mathbb{R}^m$ is linear in the Cartesian differential sense if and only if it is $\mathbb{R}$-linear in the classical sense, that is, $F(s \vec x + t \vec y) = sF(\vec x) + t F(\vec y)$ for all $s,t \in \mathbb{R}$ and $\vec x, \vec y \in \mathbb{R}^n$. So $\mathsf{Lin}[\mathsf{SMOOTH}]$ is the category of $\mathbb{R}$-linear maps between the $\mathbb{R}$-vector spaces $\mathbb{R}^n$. Similarly, a smooth function $\mathbb{R}^n \times \mathbb{R}^k \xrightarrow{G} \mathbb{R}^m$ is linear in its first argument (resp.\ second argument) if and only if it is $\mathbb{R}$-linear in its first argument $\mathbb{R}^n$ (resp.\ second argument $\mathbb{R}^k$), that is, $G(s \vec x + t \vec y, \vec z) = sG(\vec x, \vec z) + t G(\vec y, \vec z)$ (resp.\ $G(\vec z, s \vec x + t \vec y) = sG(\vec z, \vec x) + t G(\vec z, \vec x)$). Thus, $\mathbb{R}^n \times \mathbb{R}^k \xrightarrow{G} \mathbb{R}^m$ is bilinear if and only if it is $\mathbb{R}$-bilinear. 
\end{example}

We conclude this section by reviewing Cartesian closed differential categories (also sometimes called differential $\lambda$-categories \cite{bucciarelli2010categorical,manzonetto2012}). As the name suggests, these are Cartesian differential categories whose underlying category is also Cartesian closed and such that the differential combinator is compatible with the curry operator. For a Cartesian closed category, we denote the internal hom by $[A,B]$, the evaluation map by $[A,B] \times A \xrightarrow{\epsilon} B$, and for a map $A \times B \xrightarrow{f} C$, its curry is denoted $A \xrightarrow{\lambda(f)} [B,C]$. 

\begin{definition}\label{CDCcloseddef} A \textbf{Cartesian closed differential category} \cite[Section 4.6]{Cockett-2019} is a Cartesian differential category which is also a Cartesian closed category such that all evaluation maps $[A,B] \times A \xrightarrow{\epsilon} B$ are linear in their first argument (note that this implies that $\lambda(f+g) = \lambda(f) + \lambda(g)$ and $\lambda(0) = 0$). 
\end{definition}

Here are now some examples of Cartesian closed differential categories. 

\begin{example} \normalfont Every model of the differential $\lambda$-calculus \cite{ehrhard2003differential} induces a Cartesian closed differential category \cite[Theorem 4.3]{Cockett-2019}, and conversely, every Cartesian closed differential category induces a model of the differential $\lambda$-calculus \cite[Theorem 4.12]{bucciarelli2010categorical}. 
\end{example}

\begin{example} \normalfont $R\text{-}\mathsf{MOD}$ and $\mathsf{SMOOTH}$ are not Cartesian closed differential categories since neither is Cartesian closed. 
\end{example}

\begin{example} \normalfont A differential storage category \cite[Definition 10]{Blute2019} is a differential category with finite products and Seely isomorphisms, that is, $!(A \times B) \cong \oc A \otimes \oc B$ and $\oc \top \cong k$ are isomorphisms. By \cite[Theorem 4.4.2]{blute2015cartesian}, for a differential storage category whose base category is also symmetric monoidal closed (with internal homs denoted by $A \multimap B$), the coKleisli category is a Cartesian closed differential category where $[A,B] = \oc A \multimap B$. Examples can be found in \cite[Section 5]{bucciarelli2010categorical}.
\end{example}

\begin{example} \normalfont The category of convenient vector spaces and smooth functions between them is a Cartesian closed differential category \cite[Example 2.4.2]{manzyuk2012tangent}. For a detailed introduction to convenient vector spaces, see \cite{kriegl1997convenient}. Briefly, a convenient vector space is a locally convex vector space $E$ such that for every smooth function $\mathbb{R} \xrightarrow{c} E$, called a smooth curve, there exists another smooth curve $\mathbb{R} \xrightarrow{\overline{c}} E$ such that $\overline{c}^\prime = c$. If $E$ and $F$ are convenient vector spaces, then a smooth function is a function $E \xrightarrow{f} F$ which preserves smooth curves, that is, if $c$ is a smooth curve of $E$, then $f \circ c$ is a smooth curve of $F$. Let $\mathsf{CON}$ be the category of convenient vector spaces and smooth functions between them. $\mathsf{CON}$ is a Cartesian closed differential category where the internal hom is given by $[E,F] = \lbrace E \xrightarrow{f} F \vert ~ \text{$f$ is smooth} \rbrace$, which is indeed a convenient vector space, and for a smooth function $E \xrightarrow{f} F$, its derivative $E \times E \xrightarrow{\mathsf{D}[f]} F$ is defined as $\mathsf{D}[f](x, y) := \lim \limits_{t \to 0} \frac{f(x + ty) - f(x)}{t}$, where $t \in \mathbb{R}$. Note that while $\mathsf{SMOOTH}$ is not Cartesian closed, $\mathsf{SMOOTH}$ is a sub-Cartesian differential category of $\mathsf{CON}$. Furthermore, $\mathsf{CON}$ is also the coKleisli category of a differential storage category \cite{blute2010convenient}.  
\end{example}

\section{Linearly Closed Cartesian Differential Categories and Jacobians}

In this section, we introduce the main contribution of this paper: linearly closed Cartesian differential categories, which are the appropriate setting for defining Jacobians (in the Fr{\'e}chet/Gateaux derivative sense) in the context of Cartesian differential categories. 

\begin{definition} A \textbf{linearly closed Cartesian differential category} is a Cartesian differential category such that for each pair of objects $A$ and $B$, there is an object $\mathcal{L}(A,B)$, called the \textbf{internal linear hom}, and a bilinear map $\mathcal{L}(A,B) \times A \xrightarrow{\epsilon_\ell} B$, called the \textbf{evaluation map}, such that for every map $A \times B \xrightarrow{f} C$ which is linear in its second argument, there exists a unique map ${A \xrightarrow{\lambda_\ell(f)} \mathcal{L}(B,C)}$, called the \textbf{linear curry} of $f$, such that $f = A \times B \xrightarrow{\lambda_\ell(f) \times 1_B} \mathcal{L}(B,C) \times B \xrightarrow{\varepsilon_\ell} C$. 
\end{definition}

It is worth highlighting the major differences between being linearly closed and Cartesian closed. (1) $\mathcal{L}(A,B)$ should be interpreted as the internal version of $\mathsf{Lin}[\mathbb{X}](A,B)$, the linear maps from $A$ to $B$, while $[A,B]$ is the internal version of $\mathbb{X}(A,B)$, all maps from $A$ to $B$. (2) The evaluation map $\mathcal{L}(A,B) \times A \xrightarrow{\epsilon_\ell} B$ is bilinear, while $[A,B] \times A \xrightarrow{\epsilon} B$ is only linear in its first argument. The difference here is due to the fact that linear maps are additive, while an arbitrary map may not be. (3) In a linearly closed setting, we are only able to curry linear arguments. (4) It is possible to be linearly closed and not Cartesian closed, and vice versa. We stress that being linearly closed does not imply that $\mathsf{Lin}[\mathbb{X}]$ is Cartesian closed. However, if a Cartesian differential category $\mathbb{X}$ has a monoidal representation in the sense of \cite[Section 3.2]{blute2015cartesian}, then we conjecture that being linearly closed is equivalent to $\mathsf{Lin}[\mathbb{X}]$ being monoidal closed (this approach has been studied by Gallagher and MacAdam \cite{benandjon}, and we note that both $\mathsf{SMOOTH}$ and $\mathsf{CON}$ have monoidal representation). 

Now recall that {\bf [CD.6]} says that for any map $f$, $\mathsf{D}[f]$ is linear in its second argument. Therefore, the Jacobian of $f$ is defined as the linear curry of its derivative: 

\begin{definition} In a linearly closed Cartesian differential category, the \textbf{Jacobian} of a map $A \xrightarrow{f} B$ is the map $A \xrightarrow{\mathbf{J}(f)} \mathcal{L}(A,B)$ defined as the linear curry of $A \times A \xrightarrow{\mathsf{D}[f]} B$, that is, $\mathbf{J}(f) : = \lambda_\ell( \mathsf{D}[f] )$. 
\end{definition}

Here are now some examples of linearly closed Cartesian differential categories and Jacobians: 

\begin{example} \normalfont $R\text{-}\mathsf{MOD}$ is not a linearly closed Cartesian differential category since the only bilinear maps are zero maps. The same is true for any category with finite biproducts. 
\end{example}

\begin{example} \normalfont $\mathsf{SMOOTH}$ is a linearly closed Cartesian differential category where the internal linear hom is $\mathcal{L}(\mathbb{R}^n, \mathbb{R}^m) = \mathbb{R}^{nm}$ and the evaluation map ${\mathbb{R}^{nm} \times \mathbb{R}^n \xrightarrow{\epsilon_\ell} \mathbb{R}^m}$ is defined as laying out the first argument, which is a column vector of size $nm$, into a matrix of size $m \times n$ and then multiplying by the second argument, seen as a matrix of size $n \times 1$: 
\[{\small\epsilon_\ell \left( \vec x, \vec y \right) := \begin{bmatrix} x_1 & x_2 & \hdots & x_n \\
x_{n+1}& x_{n+2} & \hdots & x_{2n} \\
 \vdots & \vdots & \vdots & \vdots \\
x_{(m-1)n + 1} & x_{(m-1)n + 2} & \hdots & x_{mn}
\end{bmatrix} \begin{bmatrix} y_1 \\ y_2 \\ \vdots \\ y_n
\end{bmatrix} = \begin{bmatrix} \sum \limits^n_{i=1} x_{i} y_i \\ \vdots \\ \sum \limits^n_{i=1} x_{(m-1)n + i}y_i
\end{bmatrix}} \]
When $m = 1$, the evaluation map is given by the dot product of vectors $\epsilon_\ell( \vec x, \vec y) = \vec x \cdot \vec y = \sum \limits^n_{i=1} x_{i} y_i$. To define the linear curry, let $e_i \in \mathbb{R}^n$ be the canonical basis vectors, $e_i = [0, \hdots, 0, 1, 0, \hdots, 0]^\mathsf{T}$. Then for a smooth function $\mathbb{R}^n \times \mathbb{R}^k \xrightarrow{G = \langle g_1, \hdots, g_m \rangle} \mathbb{R}^m$, which is linear in its second argument, define $\mathbb{R}^n \xrightarrow{\lambda(G)} \mathbb{R}^{km}$ as $\lambda(G)(\vec x) = [g_1(\vec x, e_1), \hdots, g_1(\vec x, e_k),g_2(\vec x, e_1), \hdots, g_m(\vec x, e_k)]^\mathsf{T}$. For a smooth function ${\mathbb{R}^n \xrightarrow{F = \langle f_1, \hdots, f_m \rangle} \mathbb{R}^m}$, taking $\mathbf{J}(F)(\vec x) = \lambda_\ell( \mathsf{D}[F] )(\vec x)$ results precisely in the Jacobain matrix of $F$ at $\vec x$ interpreted as column vector of size $nm$, $\mathbf{J}(F)(\vec x) = [\frac{\partial f_1}{\partial x_1}(\vec x), \hdots, \frac{\partial f_1}{\partial x_n}(\vec x) ,\frac{\partial f_2}{\partial x_1}(\vec x), \hdots, \frac{\partial f_m}{\partial x_n}(\vec x)]^\mathsf{T}$, which when post-composed in the evaluation map results in laying it out back into an $m \times n$ matrix.
\end{example}

\begin{example} \normalfont For a differential category with finite products, if the base category is symmetric monoidal closed, then the coKleisli category will be a linearly closed differential category where the internal linear hom is given by $\mathcal{L}(A,B) : = A \multimap B$. For a coKleisli map $\oc A \xrightarrow{f} B$, its Jacobian $\oc A \xrightarrow{\mathbf{J}(f)} A \multimap B$ is define as the curry in the base symmetric monoidal category of the composite $\oc A \otimes A \xrightarrow{d} \oc A \xrightarrow{f} B$. Note that unlike in the Cartesian closed case, we do not need to assume the presence of Seely isomorphisms for the linearly closed case. That said, it can be shown that the coKleisli category of a differential storage category is linearly closed if and only if the base category is symmetric monoidal closed. Furthermore, this internal linear hom captures the notion of linear types in V\'ak\'ar's work on automatic differentiation \cite{vakar2021chad}. 
\end{example}

\begin{example} \normalfont $\mathsf{CON}$ is a linearly closed Cartesian differential category where the internal linear hom is $\mathcal{L}(E, F) = \lbrace E \xrightarrow{f} F \vert ~ \text{$f$ is smooth and $\mathbb{R}$-linear} \rbrace \subset [E,F]$, and the evaluation map $\mathcal{L}(E,F) \times E \xrightarrow{\epsilon_\ell} F$ is defined in the obvious way, $\epsilon_\ell(f, x) = f(x)$. A smooth function $E \times E^\prime \xrightarrow{g} F$ is linear in its second argument if and only if $g$ is $\mathbb{R}$-linear in its second argument. So if $g$ is linear in its second argument, then $E \xrightarrow{\lambda_\ell(g)} \mathcal{L}(E^\prime, F)$ is defined as $\lambda_\ell(g)(x)(y) = g(x,y)$. For a smooth function $E \xrightarrow{f} F$, its Jacobian $E \xrightarrow{\mathbf{J}(f)} \mathcal{L}(E, F)$ is defined as $\mathbf{J}(f)(x)(y) = \mathsf{D}[f](x,y)$. 
\end{example}

Being linearly closed shares many similar looking properties to being Cartesian closed. In term logic notation, we write $\lambda_\ell(f)(a) = \lambda_\ell y. f(a,y)$ and $\epsilon_\ell(g, a) = g(a)$, and $\left( \lambda_\ell y. f(a,y) \right) (b) = f(a,b)$, and so $\mathbf{J}(f)(a) = \lambda_\ell y. \frac{\mathsf{d}f(x)}{\mathsf{d}x}(a) \cdot y$. 

\begin{proposition}\label{LCCDCprop} Let $\mathbb{X}$ be a linearly closed Cartesian differential category. 
\begin{enumerate}
\item $\lambda_\ell(f + g) = \lambda_\ell(f) + \lambda_\ell(g)$ and $\lambda_\ell(0) = 0$.
\item $\mathsf{D}[\lambda_\ell(f)] = \lambda_\ell\left( \mathsf{D}[f] \circ \left \langle \langle \pi_0 \circ \pi_0, 0 \rangle, \langle \pi_1 \circ \pi_0, \pi_1 \rangle \right \rangle \right)$ (i.e.\ $\frac{\mathsf{d}\lambda_\ell y.f(x,y)}{\mathsf{d}x}(a) \cdot b = \lambda_\ell y. \frac{\mathsf{d}f(x,y)}{\mathsf{d}x}(a) \cdot c$). 
\item There is a functor $\mathsf{LIN}[\mathbb{X}]^{op} \times \mathsf{LIN}[\mathbb{X}] \xrightarrow{\mathcal{L}} \mathsf{LIN}[\mathbb{X}]$ which maps a pair of object to the internal linear hom $\mathcal{L}(A,B)$, and sends a pair of linear maps $A \xrightarrow{f} B$ and $X \xrightarrow{g} Y$ to the linear map $\mathcal{L}(B,X) \xrightarrow{\mathcal{L}(f,g)} \mathcal{L}(A,Y)$ defined as the linear curry of $\mathcal{L}(B,X) \times A \xrightarrow{1 \times f} \mathcal{L}(B,X) \times B \xrightarrow{\epsilon_\ell} X \xrightarrow{g} Y$. 
\item There are natural isomorphisms: (i) $\mathcal{L}(A, B \times C) \cong \mathcal{L}(A, B) \times \mathcal{L}(A, C)$, (ii) $\mathcal{L}(A, C) \times \mathcal{L}(B, C) \cong \mathcal{L}(A \times B, C)$, (iii) $\mathcal{L}(A, \top) \cong \top \cong \mathcal{L}(\top, A)$, and (iv) $\mathcal{L}\left( A, \mathcal{L}(B,C) \right) \cong \mathcal{L}\left( B, \mathcal{L}(A,C) \right)$. 
\end{enumerate}
\end{proposition} 
\begin{proof} (Sketch) The computations for (1) and (2) are essentially the same as those in \cite[Lemma 4.10]{Cockett-2019}. The proofs that $\mathcal{L}$ is a functor, $\mathcal{L}(A, \top) \cong \top$, $\mathcal{L}(A, B \times C) \cong \mathcal{L}(A, B) \times \mathcal{L}(A, C)$, and $\mathcal{L}\left( A, \mathcal{L}(B,C) \right) \cong \mathcal{L}\left( B, \mathcal{L}(A,C) \right)$ are essentially the same as in the Cartesian closed case. On the other hand, $\top \cong \mathcal{L}(\top, A)$ and $\mathcal{L}(A, C) \times \mathcal{L}(B, C) \cong \mathcal{L}(A \times B, C)$ intuitively follow from the fact that $\top$ is an initial object and $\times$ is a coproduct in $\mathsf{LIN}[\mathbb{X}]$. In other words, the only linear map of type $\top \to A$ is $0$, while a linear map of type $A \times B \to C$ is actually a pair of linear maps $A \to C$ and $B \to C$. 
\end{proof}

We conjecture that being linearly closed induces a non-unital closed category structure on $\mathsf{LIN}[\mathbb{X}]$. The non-unital part comes from the fact that there may not be an object $R$ such that $A \cong \mathcal{L}(R,A)$ for all objects $A$. While both $\mathsf{SMOOTH}$ and $\mathsf{CON}$ have such an object, namely $R = \mathbb{R}$, if we take $\mathsf{SMOOTH}_{even}$ the subcategory of Euclidean spaces of even dimension, then $\mathsf{SMOOTH}_{even}$ is still a linearly closed Cartesian differential category, but does not have a unit. In future work, it will be of interest to study linearly closed Cartesian differential categories with both a unit and monoidal representation. Furthermore, note that in the Cartesian closed setting we have that $[A, [B,C]] \cong [A \times B,C]$, but in the linearly closed setting, $\mathcal{L}\left( A, \mathcal{L}(B,C) \right)$ represents the bilinear maps $A \times B \to C$, which is different from $\mathcal{L}\left( A \times B, C \right)$, which represents the linear maps $A \times B \to C$, and so $\mathcal{L}\left( A, \mathcal{L}(B,C) \right) \ncong \mathcal{L}\left( A \times B, C \right)$. 

Next we turn our attention to properties of the Jacobian. In particular, we provide analogues of three basic classical identities: (1) $\mathbf{J}(F+G)(\vec x) = \mathbf{J}(F)(\vec x) + \mathbf{J}(G)(\vec x)$, (2) if $F$ is $\mathbb{R}$-linear with associated matrix $A$ (i.e.\ $F(\vec x) = A \vec x$) then $\mathbf{J}(F)(\vec x) = A$, and (3) $\mathbf{J}(G \circ F)(\vec x) = \mathbf{J}(G)(F(\vec x)) \mathbf{J}(F)(\vec x)$. The first identity is easy to generalize, while the latter two requires defining some extra maps. For a linear map $A \xrightarrow{f} B$, let $\top \xrightarrow{p_f} \mathcal{L}(A,B)$ be the linear curry of the composite $\top \times A \xrightarrow{\pi_1} A \xrightarrow{f} B$, so in term logic notation $p_f= \lambda_\ell x. f(x)$. Next, define the map ${\mathcal{L}(B,C) \times \mathcal{L}(A,B) \xrightarrow{\odot} \mathcal{L}(A,C)}$ as the linear curry of the composite ${\mathcal{L}(B,C) \times \mathcal{L}(A,B) \times A \xrightarrow{1_{\mathcal{L}(B,C)} \times \varepsilon_\ell} \mathcal{L}(B,C) \times B \xrightarrow{\varepsilon_\ell} C}$, and note that $\odot$ captures composition of linear maps, which we write in term logic notation as $g \odot f = \lambda_\ell x. g(f(x))$. We leave it to the reader to check for themselves that in $\mathsf{SMOOTH}$, if ${\mathbb{R}^n \xrightarrow{F} \mathbb{R}^m}$ is an $\mathbb{R}$-linear map with associated $n\times m$ matrix $A$, then $p_F$ is precisely $A$ in column vector form, while composition of $\mathbb{R}$-linear maps corresponds to matrix multiplication, and so $\odot$ plays the role of matrix multiplication. 

\begin{proposition}\label{propjacob} In a linearly closed Cartesian differential category: 
\begin{enumerate}[{\em (i)}]
\item $\mathbf{J}(f+g) = \mathbf{J}(f) + \mathbf{J}(g)$ and $\mathbf{J}(0) = 0$ (i.e.\ $\mathbf{J}(f+g)(a) = \mathbf{J}(f)(a) + \mathbf{J}(g)(a)$ and $\mathbf{J}(0)(a) = 0$). 
\item If $A \xrightarrow{f} B$ is linear, then $\mathbf{J}(f) = p_f \circ 0$ (i.e.\ $\mathbf{J}(f)(a) = \lambda_\ell x. f(x)$). 
\item $\mathbf{J}(g \circ f) = \odot \circ \langle \mathbf{J}(g) \circ f, \mathbf{J}(f) \rangle$ (i.e.\ $\mathbf{J}(g \circ f)(a) = \mathbf{J}(g)(f(a)) \odot \mathbf{J}(f)(a)$) 
\end{enumerate}
\end{proposition} 
\begin{proof} (Sketch) The first identity follows immediately from \textbf{[CD.1]} and Proposition \ref{LCCDCprop}.(1). For the second identity, if $f$ is linear then using the term logic we compute: $\mathbf{J}(f)(a) = \lambda_\ell y. \frac{\mathsf{d}f(x)}{\mathsf{d}x}(a) \cdot y = \lambda_\ell y. f(y)$. Lastly for the third identity, we use the chain rule \textbf{[CD.5]} to compute: $\mathbf{J}(g \circ f)(a) = \lambda_\ell y. \frac{\mathsf{d}g(f(x))}{\mathsf{d}x}(a) \cdot y = \lambda_\ell y. \frac{\mathsf{d}g(z)}{\mathsf{d}z}(f(a)) \cdot \left( \frac{\mathsf{d}f(x)}{\mathsf{d}x}(a) \cdot y \right) = \lambda_\ell y. \mathbf{J}(g)(f(a))\left( \mathbf{J}(f)(a)(y) \right) = \mathbf{J}(g)(f(a)) \odot \mathbf{J}(f)(a)$. \end{proof}

We conjecture that it is possible to provide an equivalent alternative axiomatization of a linearly closed Cartesian differential category as a Cartesian left additive category equipped with $\mathcal{L}(-,-)$, $\varepsilon_\ell$, and $\mathbf{J}$, where one would define the differential combinator as $\mathsf{D}[-] = \varepsilon_\ell \circ (\mathbf{J}(-) \times 1)$. While this is definitely of interest, it does require a bit of work to properly set up everything. As such, this will be a story for another time. 

We finish this section by providing necessary and sufficient conditions for when a Cartesian closed differential category is also linearly closed, and then explaining how every Cartesian closed differential category embeds into a linearly closed Cartesian differential category via splitting linear idempotents. We start with the definition of linear idempotents splitting linearly (which is a special case of \cite[Section 4.5]{Cockett-2019}): 

\begin{definition} In a Cartesian differential category, a \textbf{linear split idempotent} is an idempotent $A \xrightarrow{e} A$ which is linear and such that there are linear maps $A \xrightarrow{r} B$ and $B \xrightarrow{s} A$ such that $s \circ r = e$ and $r \circ s = 1_B$. A \textbf{linear idempotent complete} Cartesian differential category is a Cartesian differential category such that all linear idempotents are linear split idempotents. 
\end{definition}

The key idea is that we will define a linear idempotent on the internal hom $[A,B]$ which linearly splits via the internal linear hom $\mathcal{L}(A,B)$. To define this idempotent, we will require the notion of partial linearization \cite[Proposition 5.5]{cockett2020linearizing}. So for a map $A \times B \xrightarrow{f} C$, its linearization in context $A$ is the map $A \times B \xrightarrow{\mathsf{L}^A[f]} C$ defined as $\mathsf{L}^A[f] := \mathsf{D}[f] \circ \langle \langle \pi_0, 0 \rangle, \langle 0, \pi_1 \rangle \rangle$, i.e.\ $\mathsf{L}^A[f](a,b) = \frac{\mathsf{d}f(a,y)}{\mathsf{d}y}(0) \cdot b$. The map $\mathsf{L}^A[f]$ is now linear in its second argument. So in a Cartesian closed differential category, define ${[A,B] \xrightarrow{\ell} [A,B]}$ as the linear curry of the partial linearization of the evaluation map ${[A,B] \times A \xrightarrow{\mathsf{L}^{[A,B]}\left[ \epsilon \right] } B}$, that is, $\ell = \lambda_\ell\left( \mathsf{L}^{[A,B]}\left[ \epsilon \right] \right)$. In term logic notation, we have that $\ell(f) = \lambda_\ell y. \frac{\mathsf{d}f(x)}{\mathsf{d}x}(0) \cdot y$, and so $\ell$ is understood as mapping an arbitrary map $f$ to its linearization \cite[Proposition 3.6]{cockett2020linearizing}. As such, it follows that $\ell$ is a linear idempotent. 

\begin{proposition} A Cartesian closed differential category is linearly closed if and only if $[A,B] \xrightarrow{\ell} [A,B]$ is a linear split idempotent. 
\end{proposition} 
\begin{proof} (Sketch) Suppose that our Cartesian closed differential category is also linearly closed. As explained above, $\ell$ is already idempotent and linear. So it remains to construct a linear splitting of $\ell$. First consider the curry of the evaluation map for the linear closed structure, $\mathcal{L}(A,B) \xrightarrow{\lambda(\epsilon_\ell)} [A,B]$, which is linear since $\epsilon_\ell$ was linear in its first argument. Next, consider the linear curry of the partial linearization of the evaluation map for the Cartesian closed structure, $[A,B] \xrightarrow{\lambda_\ell \left( \mathsf{L}^{[A,B]}[\epsilon] \right)} \mathcal{L}(A,B)$, which is also linear since $\mathsf{L}^{[A,B]}[\epsilon]$ is linear in its first argument. Then it follows that $\lambda_\ell \left( \mathsf{L}^{[A,B]}[\epsilon] \right) \circ \lambda(\epsilon_\ell) = 1_{\mathcal{L}(A,B)}$ and $\lambda(\epsilon_\ell) \circ \lambda_\ell \left( \mathsf{L}^{[A,B]}[\epsilon] \right) = \ell$. So $\ell$ is a linear split idempotent. Conversely, suppose that $[A,B] \xrightarrow{\ell} [A,B]$ is a linear split idempotent via $[A,B] \xrightarrow{r} \mathcal{L}(A,B)$ and $\mathcal{L}(A,B) \xrightarrow{s} [A,B]$. Then the internal linear hom is given by the object of the splitting $\mathcal{L}(A,B)$, while the evaluation map is defined as the composite $\epsilon_\ell := \mathcal{L}(A,B) \times A \xrightarrow{s \times 1_A} [A,B] \times A \xrightarrow{\mathsf{L}^{[A,B]}[\epsilon]} B$, which is bilinear since $\mathsf{L}^{[A,B]}[\epsilon]$ is bilinear and $s$ is linear. The linear curry of a map $A \times B \xrightarrow{f} C$ which is linear in its second argument is defined as the composite $\lambda_\ell(f) := A \xrightarrow{\lambda(f)} [A,B] \xrightarrow{r} \mathcal{L}(A,B)$. So our Cartesian closed differential category is linearly closed. 
\end{proof} 

In a Cartesian closed differential category which is also linearly closed, for a map $A \times B \xrightarrow{f} C$ which is linear in its second argument, its curry $A \xrightarrow{\lambda(f)} [B,C]$ factors through its linear curry $A \xrightarrow{\lambda_\ell(f)} \mathcal{L}(B,C)$ in the sense that $\lambda(f) = A \xrightarrow{\lambda_\ell(f)} \mathcal{L}(B,C) \xrightarrow{s} [A,B]$. It is also worth pointing out that while it is true that $[A,B] \times A \xrightarrow{\mathsf{L}^{[A,B]}\left[ \epsilon \right] } B$ is bilinear and that if $A \times B \xrightarrow{f} C$ is linear in its second argument then we also have that $f = \mathsf{L}^{[B,C]}\left[ \epsilon \right] \circ (\lambda(f) \times 1_B)$, this does not give a linearly closed structure since $\lambda(f)$ may fail the uniqueness requirement (i.e.\ $\mathsf{L}^{[A,B]}\left[ \epsilon \right]$ is not monic in its first argument). For example, we have that $\mathsf{L}^{[A,B]}\left[ \epsilon \right] = \mathsf{L}^{[A,B]}\left[ \epsilon \right] \circ (\ell \times 1_A)$ but $\ell \neq 1_{[A,B]}$. 

Given any Cartesian differential category, we may build its linear idempotent completion in the obvious way. For a Cartesian differential category $\mathbb{X}$, let $\mathsf{LS}[\mathbb{X}]$ be the category whose objects are pairs $(A, A \xrightarrow{e} A)$ consisting of an object $A$ and a linear idempotent $e$ of $\mathbb{X}$ and whose maps $(A,e) \xrightarrow{f} (B, e^\prime)$ are maps $A \xrightarrow{f} B$ $\mathbb{X}$ such that $f \circ e = e^\prime \circ f$. By \cite[Proposition 4.7, Corollary 4.1]{cockett2020linearizing}, $\mathsf{LS}[\mathbb{X}]$ is a Cartesian differential category where the differential combinator is defined as in $\mathbb{X}$, and $\mathsf{LS}[\mathbb{X}]$ is also linear idempotent complete. Furthermore, if $\mathbb{X}$ was a Cartesian closed differential category, then $\mathsf{LS}[\mathbb{X}]$ is also a Cartesian closed differential category. Therefore, putting all of this together, we obtain the following: 

\begin{proposition} If $\mathbb{X}$ is a Cartesian closed differential category, then $\mathsf{LS}[\mathbb{X}]$ is a linearly closed Cartesian closed differential category. 
\end{proposition}

\section{Cartesian Reverse Differential Categories and Gradients}

In most of the standard conventions of calculus, the gradient of a scalar-valued smooth function is the transpose of its Jacobian. In this section, we explain how Cartesian reverse differential categories are the appropriate setting in which to define gradients and transposes. We briefly review Cartesian reverse differential categories and invite readers to see the full story in the original paper \cite{cockettetal:LIPIcs:2020:11661}. 

\begin{definition}\label{cartdiffdef} A \textbf{Cartesian reverse differential category} is a Cartesian left additive category $\mathbb{X}$ equipped with a \textbf{reverse differential combinator} $\mathsf{R}$, which is a family of operators $\mathbb{X}(A,B) \xrightarrow{\mathsf{R}} \mathbb{X}(A \times B,A)$, where for a map $A \xrightarrow{f} B$, the resulting map $A \times B \xrightarrow{\mathsf{R}[f]} A$ is called the reverse derivative of $f$, and such that the seven axioms found in \cite[Definition 13]{cockettetal:LIPIcs:2020:11661} hold. 
\end{definition}

\begin{example} \normalfont $\mathsf{SMOOTH}$ is a Cartesian reverse differential category where for a smooth function of type $\mathbb{R}^n \xrightarrow{F} \mathbb{R}^m$, $F = \langle f_1, \hdots, f_m \rangle$, its reverse derivative $\mathbb{R}^n \times \mathbb{R}^m \xrightarrow{\mathsf{R}[F]} \mathbb{R}^n$ is defined by multiplying the transpose of the Jacobian matrix of $F$ at the first argument $\vec x$, which is an $n \times m$ matrix $\mathbf{J}(F)(\vec x)^\mathsf{T}$, with the second argument $\vec y$, viewed as an $m \times 1$ matrix:
\[{\small \mathsf{R}[F](\vec x, \vec y) := \mathbf{J}(F)(\vec x)^\mathsf{T}\vec y = \begin{bmatrix} \frac{\partial f_1}{\partial x_1}(\vec x) & \frac{\partial f_2}{\partial x_1}(\vec x) & \hdots & \frac{\partial f_m}{\partial x_1}(\vec x) \\
 \frac{\partial f_1}{\partial x_2}(\vec x) & \frac{\partial f_2}{\partial x_2}(\vec x) & \hdots & \frac{\partial f_m}{\partial x_2}(\vec x) \\
 \vdots & \vdots & \vdots & \vdots \\
 \frac{\partial f_1}{\partial x_n}(\vec x) & \frac{\partial f_2}{\partial x_n}(\vec x) & \hdots & \frac{\partial f_m}{\partial x_n}(\vec x) 
\end{bmatrix} \begin{bmatrix} y_1 \\ y_2 \\ \vdots \\ y_m
\end{bmatrix} = \begin{bmatrix} \sum \limits^m_{j=1} \frac{\partial f_j}{\partial x_1}(\vec x) y_j \\ \vdots \\ \sum \limits^m_{j=1} \frac{\partial f_j}{\partial x_n}(\vec x) y_j
\end{bmatrix} } \]
When $m=1$, for a smooth function $\mathbb{R}^n \xrightarrow{f} \mathbb{R}$, recall that its gradient at $\vec x \in \mathbb{R}^n$ is defined as $\nabla(f)(\vec x) = \mathbf{J}(F)(\vec x)^\mathsf{T} = [\frac{\partial f}{\partial x_1}(\vec x), \hdots, \frac{\partial f_1}{\partial x_n}(\vec x) ]^\mathsf{T}$. So its reverse derivative $\mathbb{R}^n \times \mathbb{R} \xrightarrow{\mathsf{R}[f]} \mathbb{R}^n$ is given by multiplying its gradient at $\vec x$ with the scalar in the second argument, $\mathsf{R}[f](\vec x, y) = \nabla(f)(\vec x) y$. \end{example}

Every Cartesian reverse differential category is also a Cartesian differential category \cite[Theorem 16]{cockettetal:LIPIcs:2020:11661}, where the differential combinator $\mathsf{D}$ is induced by the reverse differential combinator $\mathsf{R}$. Explicitly, for a map $A \xrightarrow{f} B$, note that its second order reverse derivative is of type $(A \times B) \times A \xrightarrow{\mathsf{R}\left[ \mathsf{R}[f] \right]} A \times B$. Then the derivative of $f$ is defined as the composite $\mathsf{D}[f] := A \times A \xrightarrow{\langle 1_A, 0 \rangle \times 1_A} (A \times B) \times A \xrightarrow{\mathsf{R}\left[ \mathsf{R}[f] \right]} A \times B \xrightarrow{\pi_1} B$. Conversely, every Cartesian differential category with a \textbf{contextual linear dagger} \cite[Definition 39]{cockettetal:LIPIcs:2020:11661} is a Cartesian reverse differential category \cite[Theorem 41]{cockettetal:LIPIcs:2020:11661}. Very briefly, a contextual linear dagger amounts to an operator $(-)^{\dagger[-]}$ which sends every map $A \times B \xrightarrow{f} C$ which is linear in its second argument to a map $A \times C \xrightarrow{f^{\dagger[A]}} B$ which is again linear in its second argument, and such that $(-)^{\dagger[A]}$ makes the category of maps which are linear in context $A$ into a $\dagger$-category with finite $\dagger$-biproducts. Then the reverse derivative of a map $A \xrightarrow{f} B$ is defined by the taking the dagger in context of its derivative, that is, $\mathsf{R}[f] : = A \times B \xrightarrow{\mathsf{D}[f]^{\dagger[A]}} B$. Furthermore, these constructions are inverses of each other \cite[Theorem 42]{cockettetal:LIPIcs:2020:11661}. 

In a Cartesian reverse differential category, for any map $f$, its reverse derivative $\mathsf{R}[f]$ is linear in its second argument. Therefore we define the gradient of $f$ as the linear curry of its reverse derivative: 

\begin{definition} A \textbf{linearly closed Cartesian reverse differential category} is a Cartesian reverse differential category whose induced Cartesian differential structure is linearly closed. In a linearly closed Cartesian reverse differential category, the \textbf{gradient} of a map $A \xrightarrow{f} B$ is the map $A \xrightarrow{\nabla(f)} \mathcal{L}(B,A)$ defined as the linear curry of $A \times B \xrightarrow{\mathsf{R}[f]} A$, that is, $\nabla(f) : = \lambda_\ell( \mathsf{R}[f] )$. 
\end{definition}

Let us explain why this is the correct definition of the gradient by proving that it is the transpose of the Jacobian. To do so, we first need the notion of a transpose operator:
\begin{definition} A linearly closed Cartesian differential category has a \textbf{linear transpose} if for every pair of objects $A$ and $B$, there is a linear map $\mathcal{L}(A,B) \xrightarrow{\tau} \mathcal{L}(B,A)$ such that (i) $\tau \circ \tau = 1$, (ii) $\tau \circ p_{1_A} = p_{1_A}$ and $\tau \circ p_{\pi_j} = \tau \circ p_{\iota_j}$, and (iii) $\tau \circ \odot = \odot \circ \langle \tau \circ \pi_1, \tau \circ \pi_0 \rangle$. 
\end{definition}

As we will see below, it turns out that in the linearly closed setting, having a linear transpose is equivalent to having a contextual linear dagger (and so also to having a reverse differential combinator). Therefore, the linear transpose axioms are analogues of the fact that in a Cartesian reverse differential category, $\mathsf{Lin}[\mathbb{X}]$ is a $\dagger$-category with $\dagger$-biproducts \cite[Proposition 24]{cockettetal:LIPIcs:2020:11661}, and so $\tau(f) = f^\dagger$. Explicitly, (i) captures the fact $\dagger$ is an involution ${f^\dagger}^\dagger = f$, (ii) captures the fact that $\dagger$ maps identities to identities, $1^\dagger = 1$, and projections to injections, $\pi_j^\dagger = \iota_j$, while (iii) tells us that $\dagger$ is contravariant, $(g \circ f)^\dagger = f^\dagger \circ g^\dagger$. In $\mathsf{SMOOTH}$, if ${\mathbb{R}^n \xrightarrow{F} \mathbb{R}^m}$ is an $\mathbb{R}$-linear map with associated $n\times m$ matrix $A$, then $F^\dagger$ is the $\mathbb{R}$-linear map given by the transpose matrix $A^\mathsf{T}$, and so $\tau$ does indeed plays the role of the matrix transpose operation.

\begin{proposition} A linearly closed Cartesian reverse differential category is precisely a linearly closed Cartesian differential category equipped with a linear transpose. Furthermore: 
\begin{multicols}{2}
\begin{enumerate}[{\em (i)}]
\item $f^{\dagger[A]} = \varepsilon_\ell^{\dagger[\mathcal{L}(A,B)]} \circ (\lambda_\ell(f) \times 1)$.
\item $\lambda_\ell( f^{\dagger[A]} ) = \tau \circ \lambda_\ell(f)$. 
\item $\nabla(f) = \tau \circ \mathbf{J}(f)$.
\item $\nabla(f+g) = \nabla(f) + \nabla(g)$ and $\nabla(0) =0$.
\item If $f$ is linear then $\nabla(f) = p_{f^\dagger} \circ 0$. 
\item $\nabla(g \circ f) = \odot \circ \langle \nabla(f), \nabla(g) \circ f \rangle$. 
\end{enumerate}
\end{multicols}
\end{proposition} 
\begin{proof} (Sketch) Starting with a linearly closed Cartesian reverse differential category, first consider the contextual linear dagger of the evaluation map $\mathcal{L}(A,B) \times B \xrightarrow{\varepsilon_\ell^{\dagger[\mathcal{L}(A,B)]} } A$, and define the transpose map as its linear curry $\mathcal{L}(A,B) \xrightarrow{\tau = \lambda_\ell\left( \varepsilon_\ell^{\dagger[\mathcal{L}(A,B)]} \right)} \mathcal{L}(B,A)$. Then $\tau$ is a linear transpose since, as explained above, the linear transpose axioms will follow from the fact that $\mathsf{Lin}[\mathbb{X}]$ is a $\dagger$-category with $\dagger$-biproducts. Conversely, starting with a linearly closed Cartesian differential category with a linear transpose $\tau$, it suffices to construct a contextual linear dagger. For a map $A \times B \xrightarrow{f} C$ which is linear in its second argument, we define its dagger as $f^{\dagger[A]} := A \times C \xrightarrow{\lambda(f) \times 1_C} \mathcal{L}(B,C) \times C \xrightarrow{\tau \times 1_C} \mathcal{L}(C,B) \times C \xrightarrow{\varepsilon_\ell} B$. So we obtain a linearly closed Cartesian reverse differential category. Next, (i) follows from the fact that the contextual linear dagger preserves modification to the context, and so: $f^{\dagger[A]} = \left( \varepsilon_\ell \circ (\lambda_\ell(f) \times 1) \right)^{\dagger[A]} = \varepsilon_\ell^{\dagger[\mathcal{L}(A,B)]} \circ (\lambda_\ell(f) \times 1)$. Then (ii) is simply the linear curry of (i). For (iii), using (ii) and $\mathsf{R}[f] = \mathsf{D}[f]^{\dagger[A]}$, we compute that $\nabla(f) = \lambda_\ell( \mathsf{R}[f] ) = \lambda_\ell( \mathsf{D}[f]^{\dagger[A]} ) = \tau \circ \lambda_\ell( \mathsf{D}[f] ) = \tau \circ \mathbf{J}(f)$. The remaining three identities are simply the transpose versions of the three identities from Proposition \ref{propjacob} and the linear transpose axioms (and that if $A \xrightarrow{f} B$ is linear then $\tau \circ p_f = p_{f^\dagger}$, which is a special case of (ii)). 
\end{proof} 

We conjecture that an equivalent alternative axiomatization of a linearly closed Cartesian reverse differential category can be given in terms of a Cartesian left additive category equipped with $\mathcal{L}(-,-)$, $\varepsilon_\ell$, $\nabla$, and $\tau$, where one would define the reverse differential combinator as $\mathsf{R}[-] = \varepsilon_\ell \circ (\nabla(-) \times 1)$. In other future work, it should also be possible to generalize other important notions from classical differential calculus such as the divergence, the curl, the Laplacian, and the Hessian. In fact, recall that the Hessian matrix is defined as the Jacobian matrix of the gradient. Therefore, the Hessian can be defined in a linearly closed Cartesian reverse differential category as $A \xrightarrow{\mathbf{H}(f) := \mathbf{J}(\nabla(f))} \mathcal{L}\left( A, \mathcal{L}(B,A) \right)$. 

\nocite{*}
\bibliographystyle{eptcs}
\bibliography{actbib}

\begin{thebibliography}{10}
\providecommand{\bibitemdeclare}[2]{}
\providecommand{\surnamestart}{}
\providecommand{\surnameend}{}
\providecommand{\urlprefix}{Available at }
\providecommand{\url}[1]{\texttt{#1}}
\providecommand{\href}[2]{\texttt{#2}}
\providecommand{\urlalt}[2]{\href{#1}{#2}}
\providecommand{\doi}[1]{doi:\urlalt{http://dx.doi.org/#1}{#1}}
\providecommand{\bibinfo}[2]{#2}

\bibitemdeclare{article}{abadi2019simple}
\bibitem{abadi2019simple}
\bibinfo{author}{M.~\surnamestart Abadi\surnameend} \& \bibinfo{author}{G.~D.
  \surnamestart Plotkin\surnameend} (\bibinfo{year}{2019}):
  \emph{\bibinfo{title}{A simple differentiable programming language}}.
\newblock {\sl \bibinfo{journal}{Proceedings of the ACM on Programming
  Languages}} \bibinfo{volume}{4}(\bibinfo{number}{POPL}), pp.
  \bibinfo{pages}{1--28}, \doi{10.1145/3190508.3190551}.

\bibitemdeclare{inproceedings}{alvarez2019change}
\bibitem{alvarez2019change}
\bibinfo{author}{M.~\surnamestart Alvarez-Picallo\surnameend} \&
  \bibinfo{author}{C.-H.~L. \surnamestart Ong\surnameend}
  (\bibinfo{year}{2019}): \emph{\bibinfo{title}{Change actions: models of
  generalised differentiation}}.
\newblock In: {\sl \bibinfo{booktitle}{International Conference on Foundations
  of Software Science and Computation Structures}},
  \bibinfo{organization}{Springer}, pp. \bibinfo{pages}{45--61},
  \doi{10.1007/3-540-10286-8}.

\bibitemdeclare{article}{Blute2019}
\bibitem{Blute2019}
\bibinfo{author}{R.~F. \surnamestart Blute\surnameend},
  \bibinfo{author}{J.~R.~B. \surnamestart Cockett\surnameend},
  \bibinfo{author}{J.-S.~P. \surnamestart Lemay\surnameend} \&
  \bibinfo{author}{R.~A.~G. \surnamestart Seely\surnameend}
  (\bibinfo{year}{2020}): \emph{\bibinfo{title}{Differential Categories
  Revisited}}.
\newblock {\sl \bibinfo{journal}{Applied Categorical Structures}}
  \bibinfo{volume}{28}, pp. \bibinfo{pages}{171--235},
  \doi{10.1007/s10485-019-09572-y}.

\bibitemdeclare{article}{blute2006differential}
\bibitem{blute2006differential}
\bibinfo{author}{R.~F. \surnamestart Blute\surnameend},
  \bibinfo{author}{J.~R.~B. \surnamestart Cockett\surnameend} \&
  \bibinfo{author}{R.~A.~G. \surnamestart Seely\surnameend}
  (\bibinfo{year}{2006}): \emph{\bibinfo{title}{Differential categories}}.
\newblock {\sl \bibinfo{journal}{Mathematical structures in computer science}}
  \bibinfo{volume}{16}(\bibinfo{number}{06}), pp. \bibinfo{pages}{1049--1083},
  \doi{10.1017/S0960129506005676}.

\bibitemdeclare{article}{blute2009cartesian}
\bibitem{blute2009cartesian}
\bibinfo{author}{R.~F. \surnamestart Blute\surnameend},
  \bibinfo{author}{J.~R.~B. \surnamestart Cockett\surnameend} \&
  \bibinfo{author}{R.~A.~G. \surnamestart Seely\surnameend}
  (\bibinfo{year}{2009}): \emph{\bibinfo{title}{Cartesian differential
  categories}}.
\newblock {\sl \bibinfo{journal}{Theory and Applications of Categories}}
  \bibinfo{volume}{22}(\bibinfo{number}{23}), pp. \bibinfo{pages}{622--672}.

\bibitemdeclare{article}{blute2015cartesian}
\bibitem{blute2015cartesian}
\bibinfo{author}{R.~F. \surnamestart Blute\surnameend},
  \bibinfo{author}{J.~R.~B. \surnamestart Cockett\surnameend} \&
  \bibinfo{author}{R.~A.~G. \surnamestart Seely\surnameend}
  (\bibinfo{year}{2015}): \emph{\bibinfo{title}{Cartesian differential storage
  categories}}.
\newblock {\sl \bibinfo{journal}{Theory and Applications of Categories}}
  \bibinfo{volume}{30}(\bibinfo{number}{18}), pp. \bibinfo{pages}{620--686}.

\bibitemdeclare{article}{blute2010convenient}
\bibitem{blute2010convenient}
\bibinfo{author}{R.~F. \surnamestart Blute\surnameend},
  \bibinfo{author}{T.~\surnamestart Ehrhard\surnameend} \&
  \bibinfo{author}{C.~\surnamestart Tasson\surnameend} (\bibinfo{year}{2012}):
  \emph{\bibinfo{title}{A convenient differential category}}.
\newblock {\sl \bibinfo{journal}{Cahiers de Topologie et G{\'e}om{\'e}trie
  Diff{\'e}rentielle Cat{\'e}goriques}} \bibinfo{volume}{LIII}, pp.
  \bibinfo{pages}{211--232}.

\bibitemdeclare{article}{bucciarelli2010categorical}
\bibitem{bucciarelli2010categorical}
\bibinfo{author}{A.~\surnamestart Bucciarelli\surnameend},
  \bibinfo{author}{T.~\surnamestart Ehrhard\surnameend} \&
  \bibinfo{author}{G.~\surnamestart Manzonetto\surnameend}
  (\bibinfo{year}{2010}): \emph{\bibinfo{title}{Categorical models for simply
  typed resource calculi}}.
\newblock {\sl \bibinfo{journal}{Electronic Notes in Theoretical Computer
  Science}} \bibinfo{volume}{265}, pp. \bibinfo{pages}{213--230},
  \doi{10.1016/j.entcs.2010.08.013}.

\bibitemdeclare{article}{Cockett-2019}
\bibitem{Cockett-2019}
\bibinfo{author}{J.~D. \surnamestart Cockett\surnameend, J. R. B.;~Gallagher}
  (\bibinfo{year}{2019}): \emph{\bibinfo{title}{Categorical models of the
  differential $\lambda$-calculus}}.
\newblock {\sl \bibinfo{journal}{Mathematical Structures in Computer Science}},
  \doi{10.1017/S0960129519000070}.

\bibitemdeclare{article}{cockettetal:LIPIcs:2020:11661}
\bibitem{cockettetal:LIPIcs:2020:11661}
\bibinfo{author}{J.~R.~B. \surnamestart Cockett\surnameend},
  \bibinfo{author}{G.~S.~H. \surnamestart Cruttwell\surnameend},
  \bibinfo{author}{J.~D. \surnamestart Gallagher\surnameend},
  \bibinfo{author}{J.-S.~P. \surnamestart Lemay\surnameend},
  \bibinfo{author}{B.~\surnamestart MacAdam\surnameend},
  \bibinfo{author}{G.~\surnamestart Plotkin\surnameend} \&
  \bibinfo{author}{D.~\surnamestart Pronk\surnameend} (\bibinfo{year}{2020}):
  \emph{\bibinfo{title}{Reverse Derivative Categories}}.
\newblock {\sl \bibinfo{journal}{LIPIcs}}
  \bibinfo{volume}{152}(\bibinfo{number}{CSL 2020}), pp.
  \bibinfo{pages}{18:1--18:16}, \doi{10.4230/LIPIcs.CSL.2020.18}.

\bibitemdeclare{article}{cockett2020linearizing}
\bibitem{cockett2020linearizing}
\bibinfo{author}{J.~R.~B. \surnamestart Cockett\surnameend} \&
  \bibinfo{author}{J.-S.~P. \surnamestart Lemay\surnameend}
  (\bibinfo{year}{2022}): \emph{\bibinfo{title}{Linearizing Combinators}}.
\newblock {\sl \bibinfo{journal}{Theory and Applications of Categories}}
  \bibinfo{volume}{38}(\bibinfo{number}{13}), pp. \bibinfo{pages}{374--431}.

\bibitemdeclare{inproceedings}{cruttwell2020categorical}
\bibitem{cruttwell2020categorical}
\bibinfo{author}{G.~\surnamestart Cruttwell\surnameend},
  \bibinfo{author}{J.~\surnamestart Gallagher\surnameend} \&
  \bibinfo{author}{D.~\surnamestart Pronk\surnameend} (\bibinfo{year}{2020}):
  \emph{\bibinfo{title}{Categorical semantics of a simple differential
  programming language}}.
\newblock In: {\sl \bibinfo{booktitle}{Proceedings of Applied Category Theory
  2020}}.

\bibitemdeclare{inproceedings}{cruttwell2021categorical}
\bibitem{cruttwell2021categorical}
\bibinfo{author}{G.~\surnamestart Cruttwell\surnameend},
  \bibinfo{author}{B.~\surnamestart Gavranovi{\'c}\surnameend},
  \bibinfo{author}{N.~\surnamestart Ghani\surnameend},
  \bibinfo{author}{P.~\surnamestart Wilson\surnameend} \&
  \bibinfo{author}{F.~\surnamestart Zanasi\surnameend} (\bibinfo{year}{2022}):
  \emph{\bibinfo{title}{Categorical Foundations of Gradient-Based Learning}}.
\newblock In: {\sl \bibinfo{booktitle}{Programming Languages and Systems}},
  \bibinfo{publisher}{Springer International Publishing}, pp.
  \bibinfo{pages}{1--28}, \doi{10.1007/978-3-030-99336-8_1}.

\bibitemdeclare{article}{ehrhard2003differential}
\bibitem{ehrhard2003differential}
\bibinfo{author}{T.~\surnamestart Ehrhard\surnameend} \&
  \bibinfo{author}{L.~\surnamestart Regnier\surnameend} (\bibinfo{year}{2003}):
  \emph{\bibinfo{title}{The differential lambda-calculus}}.
\newblock {\sl \bibinfo{journal}{Theoretical Computer Science}}
  \bibinfo{volume}{309}(\bibinfo{number}{1}), pp. \bibinfo{pages}{1--41},
  \doi{10.1016/S0304-3975(03)00392-X}.

\bibitemdeclare{article}{benandjon}
\bibitem{benandjon}
\bibinfo{author}{J.~\surnamestart Gallagher\surnameend} \&
  \bibinfo{author}{B.~\surnamestart MacAdam\surnameend}:
  \emph{\bibinfo{title}{Discussions on Cartesian differential categories with
  linear homs.}}
\newblock {\sl \bibinfo{journal}{Unpublished.}}

\bibitemdeclare{book}{kriegl1997convenient}
\bibitem{kriegl1997convenient}
\bibinfo{author}{A.~\surnamestart Kriegl\surnameend} \& \bibinfo{author}{P.~W.
  \surnamestart Michor\surnameend} (\bibinfo{year}{1997}):
  \emph{\bibinfo{title}{The convenient setting of global analysis}}.
\newblock \bibinfo{volume}{53}, \bibinfo{publisher}{American Mathematical
  Soc.}, \doi{10.1090/surv/053/06}.

\bibitemdeclare{article}{laird2013constructing}
\bibitem{laird2013constructing}
\bibinfo{author}{J.~\surnamestart Laird\surnameend},
  \bibinfo{author}{G.~\surnamestart Manzonetto\surnameend} \&
  \bibinfo{author}{G.~\surnamestart McCusker\surnameend}
  (\bibinfo{year}{2013}): \emph{\bibinfo{title}{Constructing differential
  categories and deconstructing categories of games}}.
\newblock {\sl \bibinfo{journal}{Information and Computation}}
  \bibinfo{volume}{222}, pp. \bibinfo{pages}{247--264},
  \doi{10.1016/j.ic.2012.10.015}.

\bibitemdeclare{article}{lemay2018tangent}
\bibitem{lemay2018tangent}
\bibinfo{author}{J.-S.~P. \surnamestart Lemay\surnameend}
  (\bibinfo{year}{2018}): \emph{\bibinfo{title}{A Tangent Category Alternative
  to the Faa Di Bruno Construction}}.
\newblock {\sl \bibinfo{journal}{Theory and Applications of Categories}}
  \bibinfo{volume}{33}(\bibinfo{number}{35}), pp. \bibinfo{pages}{1072--1110}.

\bibitemdeclare{article}{manzonetto2012}
\bibitem{manzonetto2012}
\bibinfo{author}{G.~\surnamestart Manzonetto\surnameend}
  (\bibinfo{year}{2012}): \emph{\bibinfo{title}{What is a categorical model of
  the differential and the resource $\lambda$-calculi?}}
\newblock {\sl \bibinfo{journal}{Mathematical Structures in Computer Science}}
  \bibinfo{volume}{22}(\bibinfo{number}{3}), \doi{10.1017/S0960129511000594}.

\bibitemdeclare{article}{manzyuk2012tangent}
\bibitem{manzyuk2012tangent}
\bibinfo{author}{O.~\surnamestart Manzyuk\surnameend} (\bibinfo{year}{2012}):
  \emph{\bibinfo{title}{Tangent bundles in differential lambda-categories}}.
\newblock {\sl \bibinfo{journal}{arXiv preprint arXiv:1202.0411}}.

\bibitemdeclare{inproceedings}{sprunger2019differentiable}
\bibitem{sprunger2019differentiable}
\bibinfo{author}{D.~\surnamestart Sprunger\surnameend} \&
  \bibinfo{author}{S.~\surnamestart Katsumata\surnameend}
  (\bibinfo{year}{2019}): \emph{\bibinfo{title}{Differentiable causal
  computations via delayed trace}}.
\newblock In: {\sl \bibinfo{booktitle}{2019 34th Annual ACM/IEEE Symposium on
  Logic in Computer Science (LICS)}}, \bibinfo{organization}{IEEE}, pp.
  \bibinfo{pages}{1--12}, \doi{10.1109/LICS.2019.8785670}.

\bibitemdeclare{article}{vakar2021chad}
\bibitem{vakar2021chad}
\bibinfo{author}{M.~\surnamestart V{\'a}k{\'a}r\surnameend}
  (\bibinfo{year}{2021}): \emph{\bibinfo{title}{CHAD: Combinatory Homomorphic
  Automatic Differentiation}}.
\newblock {\sl \bibinfo{journal}{arXiv preprint arXiv:2103.15776}}.

\bibitemdeclare{inproceedings}{wilson2021reverse}
\bibitem{wilson2021reverse}
\bibinfo{author}{P.~\surnamestart Wilson\surnameend} \&
  \bibinfo{author}{F.~\surnamestart Zanasi\surnameend} (\bibinfo{year}{2020}):
  \emph{\bibinfo{title}{Reverse Derivative Ascent: A Categorical Approach to
  Learning Boolean Circuits}}.
\newblock In: {\sl \bibinfo{booktitle}{Proceedings of Applied Category Theory
  2020}}, \doi{10.4204/EPTCS.333.17}.

\end{thebibliography}
\end{document}